\pdfoutput=1

\documentclass[11pt]{article}

\usepackage[utf8]{inputenc}
\usepackage{listings}
\usepackage{amsmath}
\usepackage[skins, breakable]{tcolorbox}
\usepackage{multirow}
\usepackage{amsthm}
\usepackage{amsfonts}
\usepackage{appendix}
\usepackage{xspace}
\usepackage{color}
\usepackage{tikz}
\usepackage{pgfplots}
\usepackage{textcomp}
\usepackage{enumitem}
\usepackage{float}
\usepackage{hhline}
\usepackage{bm}

\newcommand{\echo}{\texttt{echo}\xspace}
\newcommand{\ack}{\texttt{ack}\xspace}
\newcommand{\propose}{\texttt{propose}\xspace}

\newcommand{\vote}{\texttt{vote}\xspace}
\newcommand{\voteone}{\texttt{vote-1}\xspace}
\newcommand{\votetwo}{\texttt{vote-2}\xspace}

\newcommand{\lock}{\texttt{lock}\xspace}

\newcommand{\latency} {good-case latency\xspace}
\newcommand{\badlatency} {bad-case latency\xspace}
\newcommand{\clockskew}{\sigma}

\newtcolorbox{mybox}[1][]{
enhanced,
colback=white,
boxsep=0pt,
#1
} 

\newcommand{\stitle}[1]{\vspace{0.5ex} \noindent\textsf{\textbf{#1}}}

\renewcommand{\paragraph}[1]{\smallskip\stitle{#1}}

\usepackage{authblk}
\usepackage{fullpage}
\usepackage{hyperref}
\newtheorem{lemma}{Lemma}
\newtheorem{theorem}{Theorem}
\newtheorem{definition}{Definition}

\title{Good-case and Bad-case Latency of Unauthenticated Byzantine Broadcast: A Complete Categorization} %TODO Please add

\author[1]{Ittai Abraham}
\author[2]{Ling Ren}
\author[3]{Zhuolun Xiang}
\affil[1]{VMware Research\\ {iabraham@vmware.com}}
\affil[2, 3]{University of Illinois at Urbana-Champaign\\ {\{renling, xiangzl\}@illinois.edu}}

\begin{document}

\maketitle

\begin{abstract}
    This paper studies the {\em good-case latency} of {\em unauthenticated} Byzantine fault-tolerant broadcast, which measures the time it takes for all non-faulty parties to commit given a non-faulty broadcaster. For both asynchrony and synchrony, we show that $n\geq 4f$ is the tight resilience threshold that separates good-case 2 rounds and 3 rounds. 
    For asynchronous Byzantine reliable broadcast (BRB), we also investigate the {\em bad-case latency} for all non-faulty parties to commit when the broadcaster is faulty but some non-faulty party commits. We provide matching upper and lower bounds on the resilience threshold of bad-case latency for BRB protocols with optimal good-case latency of 2 rounds. In particular, we show 2 impossibility results and propose 4 asynchronous BRB protocols.
\end{abstract}

\section{Introduction}
\label{sec:intro}

\begin{table*}[t]
\centering
\setlength\doublerulesep{0.5pt}
\begin{tabular}{|c|c|c|c|c|}
\hline
\textbf{Problem} & \textbf{Timing Model} & \textbf{Resilience} & \textbf{Lower Bound} &
\textbf{Upper Bound} 
\\ 
\hhline{=====} 
\multirow{2}{*}{BRB} & \multirow{2}{*}{Asynchrony} & $n\geq 4f$ & $2$ rounds~\cite{abraham2021goodcase}
& \textbf{$\bm{2}$ rounds} (Thm~\ref{thm:ub:async}) \\ 
\cline{3-5} 
& & $3f+1\leq n\leq 4f-1$ & \textbf{$\bm{3}$ rounds} (Thm~\ref{thm:lowerbound})  & $3$ rounds~\cite{bracha1987asynchronous} \\
% \\ 
%  \hhline{=====}
% \multirow{2}{*}{Psync-BB} & \multirow{2}{*}{Partial Synchrony} & $n\geq 5f-1$ & $2$ rounds~\cite{abraham2021goodcase}
% & $2$ rounds~\cite{martin2006fast} \\ \cline{3-5} 
% & & $3f+1\leq n\leq 5f-2$ & $3$ rounds~\cite{martin2006fast} & $3$ rounds~\cite{castro1999practical} 
% \\ 
\hhline{=====} 
\multirow{2}{*}{BB} & \multirow{2}{*}{Synchrony} & $n\geq 4f$ & $2\delta$~\cite{abraham2021goodcase, dolev1990early}
& \textbf{$\bm{2\delta}$} (Thm~\ref{thm:ub:sync})  \\ \cline{3-5} 
& & $3f+1\leq n\leq 4f-1$ & $\bm{3\delta}$ (Thm~\ref{thm:sync}) & $3\delta$ (Thm~\ref{thm:sync:3r})~\cite{bracha1987asynchronous} 
\\ 
 \hline
\end{tabular}
\smallskip 
\caption{Upper and lower bounds for \latency of {\em unauthenticated} Byzantine fault-tolerant broadcast. 
% \textbf{Our new results are marked bold.} \rl{Do we have to introduce delta? I think it is easier to stick with rounds.}
}
% \vspace{-1em}
\label{table:results}
\end{table*}

\begin{table*}[t]
\centering
\setlength\doublerulesep{0.5pt}
\begin{tabular}{|c|c|c|c|c|c|}
\hline
\textbf{Result} & \textbf{Resilience} & \textbf{Good-case} &
\textbf{Bad-case} & \textbf{Comm. cost} & \textbf{Reference}
\\ 
\hhline{======} 
Bracha & $n\geq 3f+1$ & 3 rounds & 4 rounds & $O(n^2)$ & \cite{bracha1987asynchronous} \\ 
\hline
Imbs and Raynal & $n\geq 5f+1$ & 2 rounds & 3 rounds & $O(n^2)$ & \cite{imbs2016trading} \\
% \hline
% {\bf Impossibility} & $n\leq 5f-2 $ & 2 rounds & 3 rounds \\ 
% \hline
% \multirow{4}{*}{{\bf Our protocols}} 
% & $n\geq 4f$ & 2 rounds & 4 rounds & $O(n^2)$ & Thm~\ref{thm:ub:async} \\
% \cline{2-6} 
% & $n\geq 4f, f=2$ & 2 rounds & 3 rounds & $O(n^3)$ & Thm~\ref{thm:f=2} \\
% \cline{2-6} 
% & $n\geq 4f, f=1$ & 2 rounds & 2 rounds & $O(n^2)$ & Thm~\ref{thm:f=1} \\
% \cline{2-6} 
% & $n\geq 5f-1$ & 2 rounds & 3 rounds & $O(n^2)$ & Thm~\ref{thm:5f-1} \\
% \hline
% \hline
\hhline{======} 
Impossibility of $(2,2)$ & $f\geq 2$ & 2 rounds & 2 rounds & - & Thm~\ref{thm:simplelb} \\
\hline
F1-BRB & $n\geq 4f, f=1$ & 2 rounds & 2 rounds & $O(n^2)$ & Thm~\ref{thm:f=1} \\
\hline
Impossibility of $(2,3)$ & $n\leq 5f-2,f\geq 3$ & 2 rounds & 3 rounds & - & Thm~\ref{thm:impossibility} \\
\hline
F2-BRB & $n\geq 4f, f=2$ & 2 rounds & 3 rounds & $O(n^3)$ & Thm~\ref{thm:f=2} \\
\hline
$(2,4)$-BRB & $n\geq 4f$ & 2 rounds & 4 rounds & $O(n^2)$ & Thm~\ref{thm:ub:async} \\
\hline
$(2,3)$-BRB & $n\geq 5f-1$ & 2 rounds & 3 rounds & $O(n^2)$ & Thm~\ref{thm:5f-1} \\
\hline
\end{tabular}
\smallskip 
\caption{Comparison of our results and previous results of asynchronous {\em unauthenticated} Byzantine reliable broadcast.}
% \vspace{-1em}
\label{table:protocols}
\end{table*}

% \IA{the two tables are great! maybe add a column that links to the Theorem inside our paper where relevant? also see suggestion to give the protocols names}

Byzantine fault-tolerant broadcast is a fundamental primitive in distributed systems, where a designated broadcaster sends its value to all parties, such that all non-faulty parties commit on the same value despite arbitrary deviation from Byzantine parties. Moreover, if the broadcaster is non-faulty, then all honest parties are required to commit the same value as the broadcaster's input. Byzantine broadcast (BB) requires all non-faulty parties to eventually commit, while Byzantine reliable broadcast (BRB) relaxes the condition to only require termination when the broadcaster is honest or if some non-faulty party terminates.
When the network is asynchronous, meaning the message delays are unbounded, it is well-known that BB is unsolvable with even a single fault. On the other hand, BRB is solvable under asynchrony as long as there are $n\geq 3f+1$ parties.

Recent work of Abraham et al.~\cite{abraham2021goodcase} investigates the notion of \latency of Byzantine fault-tolerant broadcast, which is the time for all honest parties to commit given that the broadcaster is honest. 
Theoretically, the \latency is a natural and interesting metric that has not been formally studied by the literature until recently; 
Practically, for applications like leader-based Byzantine fault-tolerant state machine replication (BFT SMR), the \latency study answers the fundamental question of how fast can leader-based BFT SMR commit decisions during the steady state when the leader is non-faulty. 
Moreover, for asynchronous Byzantine reliable broadcast, \latency is particularly important since BRB may not terminate under a Byzantine leader. 
%For instance, the classic Bracha's reliable broadcast~\cite{bracha1987asynchronous} has a \latency of 3 rounds, but no honest can commit when the broadcaster remains silent.

The work of Abraham et al.~\cite{abraham2021goodcase} reveals a surprisingly rich structure in the \latency tight bounds for {\em authenticated} Byzantine broadcast, where digital signatures are used and the adversary is assumed to be computationally bounded. 
In this work, we study the \latency and \badlatency of {\em unauthenticated} Byzantine fault-tolerant broadcast.
Our results are summarized in Table~\ref{table:results} and \ref{table:protocols}. 
%Our results answer some open questions raised by~\cite{abraham2021goodcase}.

\paragraph{Complete categorization for \latency under asynchrony and synchrony.}
Under asynchrony when the message delays are unbounded, we show that $n\geq 4f$ is the tight resilience threshold that separates \latency of 2 rounds and 3 rounds. 
For $n\geq 4f$, \cite{abraham2021goodcase} shows a 2-round lower bound, and we present a protocol with \latency of 2 rounds.
For $3f+1\leq n\leq 4f-1$, Bracha's reliable broadcast~\cite{bracha1987asynchronous} has \latency of 3 rounds, and we prove a matching 3-round lower bound.

\begin{theorem}[Informal; tight bounds on \latency in asynchrony]
For unauthenticated Byzantine reliable broadcast with $f$ Byzantine parties under asynchrony, in the good-case: 
\begin{enumerate}
\item 2 rounds are necessary and sufficient if $\, n \geq 4f $ (Section~\ref{sec:async:protocol}), and   
\item 3 rounds are necessary and sufficient if $\, 3f+1 \leq n < 4f$ (Section~\ref{sec:async:lowerbound}).
\end{enumerate}
\end{theorem}

The above asynchronous \latency bounds also imply similar results for \latency of BB and BRB under synchrony as well. Let $\delta$ denote the actual message delay bound during the execution (see Section~\ref{sec:sync} for details). 
For $n\geq 4f$, \cite{abraham2021goodcase} shows a $2\delta$ lower bound (also implied by the early-stopping results~\cite{dolev1990early}), and we present a synchronous BB protocol  with \latency of $2\delta$, inspired by our 2-round asynchronous BRB protocol.
For $3f+1\leq n\leq 4f-1$, we show a synchronous BB protocol that has \latency of $3\delta$ inspired by Bracha's reliable broadcast~\cite{bracha1987asynchronous}, and the aforementioned $3\delta$ lower bound also applies to synchrony.

\begin{theorem}[Informal; tight bounds on \latency in synchrony]
For unauthenticated Byzantine broadcast and Byzantine reliable broadcast with $f$ Byzantine parties under synchrony (message delay bounded by $\delta$), in the good-case: 
\begin{enumerate}
\item $2\delta$ are necessary and sufficient if $\, n \geq 4f $ (Section~\ref{sec:sync}), and   
\item $3\delta$ are necessary and sufficient if $\, 3f+1 \leq n < 4f$ (Section~\ref{sec:sync}).
\end{enumerate}
\end{theorem}

\paragraph{Complete categorization for \badlatency of asynchronous Byzantine reliable broadcast.}
In addition to the good-case commit path, asynchronous BRB protocols usually have a second commit path to ensure all honest parties eventually commit, when the Byzantine broadcaster and Byzantine parties deliberately make only a few honest parties commit in the good-case commit path. We use {\em \badlatency} to denote the latency of such second commit path, and say a BRB protocol is $(R_g, R_b)$-round if it has \latency of $R_g$ rounds and \badlatency of $R_b$ rounds. 
For instance, Bracha's reliable broadcast~\cite{bracha1987asynchronous} is $(3,4)$-round. 

We provide a complete categorization of the threshold resilience for BRB with good-case latency of 2. We show two lower bound results on the resilience threshold: for $(2,2)$ and for $(2,3)$. We also show 4 protocols with matching resilience bounds: these protocols have the optimal \latency of 2 rounds, but with different trade-offs in resilience and \badlatency, matching the lower bound results. 
% \IA{4 new protocols is a lot! makes it hard for the reader, maybe we can give them names?  F1 (when $f=1$) , F2, "Early Stopping RB" (this is our 2,4) and "One Echo RB" (this is the raynal variant)}
% \rl{I suggest we keep the main story "good-case" (which is also in the title). That means the 4f (early-stopping) protocol can be the main dish. The f=1 and f=2 results can be special cases of the early-stopping protocol, not as separate protocols.  }
As summarized in Table~\ref{table:protocols}, prior upper bound results include Bracha's $(3,4)$-round BRB for $n\geq 3f+1$, and the $(2,3)$-round BRB for $n\geq 5f+1$ by Imbs and Raynal~\cite{imbs2016trading}.
% We first show it is impossible to achieve $(2,2)$-round with optimal resilience $n\geq 4f$ for $f\geq 2$.
\begin{itemize}
    \item First, we show it is impossible to achieve $(2,2)$-round BRB, except for the special case of $f=1$ where we propose a protocol F1-BRB that has $(2,2)$-round and optimal resilience $n\geq 4f$.
    \item Next, we show another impossibility result stating that no BRB protocol can achieve $(2,3)$-round under $n\leq 5f-2$ for $f\geq 3$. That is, for $f\geq 3$, no BRB protocol can have optimality in all three metrics: \latency, \badlatency and resilience.
    For the special case of $f=2$, we propose a protocol F2-BRB that has $(2,3)$-round and optimal resilience $n\geq 4f$.
    For the general case of $f\geq 3$, we have two protocols -- a protocol named $(2,4)$-BRB under $n\geq 4f$ that has $(2,4)$-round, and a protocol named $(2,3)$-BRB which improves the resilience of Imbs and Raynal~\cite{imbs2016trading} from $n\geq 5f+1$ to $n\geq 5f-1$ while keeping the protocol $(2,3)$-round. Both $(2,4)$-BRB and $(2,3)$-BRB have tight resilience and latencies due to the impossibility result.
\end{itemize}

% We propose a protocol named $(2,4)$-BRB with optimal \latency of 2 rounds, optimal resilience of $n\geq 4f$, and \badlatency of 4 rounds.
% For the special case of $f=2$ and $f=1$, we construct two protocols F2-BRB and F1-BRB that have \badlatency of 3 rounds and 2 rounds, respectively, while keeping the optimal \latency and resilience.
% We also propose a protocol $(2,3)$-BRB which improves the resilience of Imbs and Raynal~\cite{imbs2016trading} from $n\geq 5f+1$ to $n\geq 5f-1$ while keeping the protocol $(2,3)$-round.
% Finally, we show it is impossible to achieve $(2,2)$-round with $n\geq 4f$ for $f\geq 2$, showing our F2-BRB is optimal; we also show it is impossible to achieve $(2,3)$-round with resilience $n\leq 5f-2$ for $f\geq 3$, showing our $(2,3)$-BRB and $(2,4)$-BRB are optimal.

% \IA{compare our results to the classic lower bounds of early stopping in synchrony of $\{t+1,f+2\}$}

\section{Preliminaries}
\label{sec:prelim}

\paragraph{Model of execution.} 
We define a protocol for a set of $n$ parties, among which at most $f$ are Byzantine faulty and can behave arbitrarily and has unbounded computational power. If a party remains non-faulty for the entire protocol execution, we call the party honest.
During an execution $E$ of a protocol, parties perform sequences of events, including {\em send, receive/deliver, local computation}. 

In this paper, we investigate results for deterministic unauthenticated protocols. If the protocol is deterministic, for any two executions, if an honest party has the same initial state and receives the same set of messages at the same corresponding time points (by its local clock), the honest party cannot distinguish two executions.
We will use the standard indistinguishability argument to prove lower bounds.

We consider both synchronous and asynchronous network models.
Under synchrony, any message between two honest parties will be delivered within $\delta$ time during the execution. 
More details about the synchrony model assumption is deferred to Section~\ref{sec:sync}.
Under asynchrony, the adversary can control the message delay of any message to be an arbitrary non-negative value.
We assume all-to-all, reliable and authenticated communication channels, such that the adversary cannot fake, modify or drop the messages sent by honest parties.

\paragraph{Byzantine broadcast variants.}
We investigate two standard variants of Byzantine broadcast problem for synchrony and asynchrony.

\begin{definition}[Byzantine Broadcast (BB)]
\label{def:bb}
A Byzantine broadcast protocol must satisfy the following properties.
    \begin{itemize}[itemsep=0pt,topsep=0pt]
        \item Agreement. If two honest parties commit values $v$ and $v'$ respectively, then $v=v'$.
        \item Validity. If the designated broadcaster is honest, then all honest parties commit the broadcaster's value and terminate.
        \item Termination. All honest parties commit and terminate.
    \end{itemize}
\end{definition}

\begin{definition}[Byzantine Reliable Broadcast (BRB)]
\label{def:brb}
A Byzantine reliable broadcast protocol must satisfy the following properties.
    \begin{itemize}[itemsep=0pt,topsep=0pt]
        \item Agreement. Same as above.
        \item Validity. Same as above.
        \item Termination. If an honest party commits a value and terminates, then all honest parties commit a value and terminate.
    \end{itemize}
\end{definition}

We will also use {\em Byzantine agreement} as a primitive to simplify the construction of our BB protocols under synchrony in Section~\ref{sec:sync}. 
The Byzantine agreement gives each party an input, and its validity requires that if all honest parties have the same input value, then all honest parties commit that value.

\paragraph{Good-case latency of broadcast.}
Depending on the network model, the measurement of latency is different. Under synchrony, we can measure the latency using the physical clock time.

\begin{definition}[Good-case Latency under Synchrony~\cite{abraham2021goodcase}]\label{def:goodcase:sync}
    A Byzantine broadcast (or Byzantine reliable broadcast) protocol has \latency of $~T$ under synchrony, if all honest parties commit within time $~T$ since the broadcaster starts the protocol (over all executions and adversarial strategies), given the designated broadcaster is honest.
\end{definition}

Under asynchrony, the network delay is unbounded.
To measure the latency of asynchronous protocols, we use the natural notion of {\em asynchronous rounds} from the literature~\cite{canetti1993fast}, where a protocol runs in $R$ asynchronous rounds if its running time is at most $R$ times the maximum message delay between honest parties during the execution.

\begin{definition}[Good-case Latency under Asynchrony~\cite{abraham2021goodcase}]\label{def:goodcase}
    A Byzantine reliable broadcast protocol has \latency of $R$ rounds under asynchrony, if all honest parties commit within asynchronous round $R$ (over all executions and adversarial strategies), given the designated broadcaster is honest.
\end{definition}

When the broadcaster is dishonest, Byzantine broadcast will have worst-case latency of $f+1$ rounds~\cite{fischer1982lower}, 
% \rl{doesn't it translate to $\delta$?}\daniel{I think it is $\Delta$} \rl{can you explain the argument?}\daniel{The lower bound proof makes f nodes crash sequentially, so these f bad nodes can send its message with $\Delta$ delay without being caught by honest nodes.}
and for Byzantine reliable broadcast by definition it does not guarantee termination (the broadcaster can just remain silent).
% \rl{Does not sound right. BRB is easier than BB. I think you are implying we also consider BRB in async}\daniel{Even for sync, BRB does not need to terminate.}
% \rl{BRB does not have to terminate, but it is incorrect to say BRB worst-latency has to be infinite. Any sync BB protocol is a sync BRB without bounded worst case latency }\daniel{Ok, I will just say it does not terminate then.}
Therefore, the notion of \latency is the natural metric to measure the latency performance of reliable broadcast. 
Another important latency metric for reliable broadcast is to measure how fast can all honest parties commit, once an honest party commit. We formally define it as the \badlatency as below.

\begin{definition}[Bad-case Latency under Asynchrony]\label{def:badcase}
    A Byzantine reliable broadcast protocol has \badlatency of $R'=R+ R_{ex}$ rounds under asynchrony, if all honest parties commit within $ R_{ex}$ asynchronous round after an honest party commits (over all executions and adversarial strategies), and the \latency of the protocol is $R$.
\end{definition}

We will use the notation $(R_g,R_b)$-round BRB to denote an authenticated Byzantine reliable broadcast protocol that has \latency or $R_g$ rounds and \badlatency of $R_b$ rounds.
For instance, the classic Bracha reliable broadcast~\cite{bracha1987asynchronous} has a \latency of $3$ rounds and a \badlatency of $4$ rounds ($ R_{ex}=1$), under $n\geq 3f+1$ parties; it is thus a $(3,4)$-round BRB.

\section{Good-case Latency of Asynchronous Byzantine Reliable Broadcast}
\label{sec:async:goodcase}
Under asynchrony, Byzantine reliable broadcast is solvable if and only if $n\geq 3f+1$.
We show the {\em tight} lower and upper bound on the \latency of asynchronous unauthenticated BRB is $2$ rounds when $n\geq 4f$, and $3$ rounds when $3f+1\leq n\leq 4f-1$.

\subsection{2-round Unauthenticated BRB under $n\geq 4f$}
\label{sec:async:protocol}
We show the tightness of the bound by presenting a $2$-round unauthenticated BRB protocol, which has \latency of $2$ rounds and \badlatency of 4 rounds with $n\geq 4f$ parties, as presented in Figure~\ref{fig:brb}. 

%\rl{Does the f+2 bound still apply? If so, we can claim even the best-case is 2 rounds. Worth checking} \daniel{It is showed by the proof of our goodcase latency paper (although we didn't claim it there).}
In the protocol, in the first round the broadcaster sends its proposal to all parties. Then in the second round, all parties send an \ack for the first proposal received. Parties commit in $2$ rounds when receiving $n-f-1$ \ack for the same value from distinct parties other than the broadcaster, which will happen when the broadcaster is honest. To ensure termination, the protocol has another 4-round commit path, to guarantee that all honest parties will commit even if the Byzantine parties deliberately make only a few honest parties commit in round $2$. The $4$-round commit path consists of a Bracha-style reliable broadcast, where the parties send \voteone and \votetwo messages upon receiving enough messages as specified in Step~\ref{rb1:step:vote}. Finally, when receiving enough \votetwo messages, party can also commit in round $4$.

\begin{figure}[tb]
    \centering
    \begin{mybox}
\begin{enumerate}%[itemsep=0pt,topsep=0pt]
    \item\label{rb1:step:propose} \textbf{Propose.} The designated broadcaster $L$  with input $v$ sends $\langle \texttt{propose}, v\rangle$ to all  parties.

    \item\label{rb1:step:ack} \textbf{Ack.} 
    When receiving the first proposal $\langle \texttt{propose}, v\rangle$ from the broadcaster,
    a party sends an $\langle \ack, v \rangle$ message to all parties.

    \item\label{rb1:step:commit1} \textbf{2-round Commit.}
    When receiving $\langle \ack, v \rangle$ from $n-f-1$ distinct non-broadcaster parties, a party commits $v$, sends $\langle \voteone, v \rangle$ and $\langle \votetwo, v \rangle$ to all parties, and terminates.

    \item\label{rb1:step:vote} \textbf{Vote.} 
    \begin{itemize}
        \item When receiving $\langle \ack, v \rangle$ from $n-2f$ distinct non-broadcaster parties, a party sends a $\langle \voteone, v \rangle$  message to all parties, if it has not already sent \voteone.
        \item When receiving $\langle \voteone, v \rangle$ from $n-f-1$ distinct non-broadcaster parties, a party sends a $\langle \votetwo, v \rangle$ message to all parties, if it has not already sent \votetwo.
        \item When receiving $\langle \votetwo, v \rangle$ from $f+1$ distinct non-broadcaster parties, a party sends a $\langle \votetwo, v \rangle$ message to all parties, if it has not already sent \votetwo.
    \end{itemize}

    \item\label{rb1:step:commit2} \textbf{4-round Commit.}
    When receiving $\langle \votetwo, v \rangle$ from $n-f-1$ distinct non-broadcaster parties, a party commits $v$ and terminates.

\end{enumerate}
    \end{mybox}
    % \vspace{-1em}
    \caption{$(2,4)$-round BRB protocol under $n\geq 4f$ 
    % \rl{shortened and added subjects to the description. I prefer this style for all codes}
    }
    % \vspace{-1em}
    \label{fig:brb}
\end{figure}

\begin{lemma}\label{lem:ub:async}
    If an honest party commits $v$ at Step~\ref{rb1:step:commit1}, then no honest party will send \voteone or \votetwo for any other value $v'\neq v$.
\end{lemma}

\begin{proof}
    Since the honest party commit $v$ at Step~\ref{rb1:step:commit1}, it receives $n-f-1$ \ack messages for $v$ from distinct non-broadcaster parties.
    If the broadcaster is honest, then no honest party will send \voteone or \votetwo message for $v'$ since there are at most $f$ Byzantine parties.
    If the broadcaster is Byzantine, and suppose there are $t$ Byzantine parties, then there are at most $t-1$ Byzantine parties among all non-broadcaster parties, and there must be at least $(n-f-1)-(t-1)=n-f-t$ honest parties sending \ack for $v$. Suppose an honest party receives $n-2f$ \ack messages for $v'$ from distinct non-broadcaster parties, then there must be at least $(n-2f)-(t-1)=n-2f-t+1$ honest parties sending \ack for $v'$. 
    Since there are only $n-t$ 
    % \rl{need to be more rigorous here. f is an upper bouund, so there may be more than 3f honest. It may be cleaner to invoke quorum intersection between two non-broadcaster quorums} 
    honest parties, there must be at least $(n-f-t)+(n-2f-t+1)-(n-t)=n-3f-t+1\geq n-4f+1\geq 1$ honest party that sends \ack for both $v$ and $v'$, contradiction. Hence no honest party can receive $n-2f$ \ack messages for $v'$. Moreover, since the thresholds in Step~\ref{rb1:step:vote} are larger than the number of Byzantine parties, i.e., $n-f-1\geq 3f-1>f$ and $f+1>f$, no honest party will send \voteone or \votetwo for $v'\neq v$.
\end{proof}

\begin{theorem}\label{thm:ub:async}
    The protocol in Figure~\ref{fig:brb} solves Byzantine reliable broadcast under asynchrony with optimal resilience $n\geq 4f$ and optimal \latency of 2 rounds, and has \badlatency of 4 rounds.
\end{theorem}

\begin{proof}
    {\bf Validity and Good-case Latency.}
    
    If the broadcaster is honest, it sends the same proposal of value $v$ to all parties. Then all $n-f-1$ non-broadcaster honest parties will multicast the \ack message for $v$. The Byzantine parties cannot make any honest party to send \voteone, \votetwo, for any other value $v'\neq v$ since $f$ is below any threshold specified in the protocol. All honest parties will eventually commit $v$ after receiving $n-f-1$ \ack messages at Step~\ref{rb1:step:commit1} and terminate. 
    The \latency is $2$ rounds, including broadcaster sending the proposal and all parties sending \ack message.
    
    {\bf Agreement.}
    
    If the broadcaster is honest, by validity all honest parties will commit the same value. Now consider when the broadcaster is Byzantine, there are at most $f-1$ Byzantine parties among non-broadcasters.
    
    If any two honest parties commit different values at Step~\ref{rb1:step:commit1}, then there must be at least $n-f-1-(f-1)=n-2f\geq 2f$ honest parties sending \ack for each of these different values. It is impossible by quorum intersection since there are only $3f$ honest parties. Similarly, no two honest parties can commit different values at Step~\ref{rb1:step:commit2}.
    
    Now we show that if an honest party $h1$ commits $v$ at Step~\ref{rb1:step:commit1} and another honest party $h2$ commits $v'$ at Step~\ref{rb1:step:commit2}, then it must be $v=v'$. 
    Suppose $h1$ commits $v$ at Step~\ref{rb1:step:commit1}, then by Lemma~\ref{lem:ub:async}, no honest party will send \voteone or \votetwo for $v'\neq v$, and thus not enough \votetwo for any $v'\neq v$ to be committed at Step~\ref{rb1:step:commit2}.
    Suppose $h2$ commit $v'$ at Step~\ref{rb1:step:commit2}, then $h2$ receives at least $n-f-1-(f-1)=n-2f\geq 2f$ \votetwo messages for $v'$ from honest parties. By the contrapositive of Lemma~\ref{lem:ub:async}, no honest party commits $v\neq v'$ at Step~\ref{rb1:step:commit1}.
    
    {\bf Termination and Bad-case Latency.}
    
    If the broadcaster is honest, by validity all honest parties will commit the same value. Now consider when the broadcaster is Byzantine, there are at most $f-1$ Byzantine parties among non-broadcasters.
    
    Suppose that an honest party commits $v$ at Step~\ref{rb1:step:commit1}, by Lemma~\ref{lem:ub:async}, no honest party will send \voteone or \votetwo for any $v'\neq v$. Since there are at least $n-f-1-(f-1)=n-2f$ non-broadcaster honest parties sending \ack for $v$, all honest parties will receive at least $n-2f$ \ack for $v$ from non-broadcasters, and thus send \voteone for $v$. Since there are $n-f$ honest non-broadcasters, all honest parties will send \votetwo for $v$, and then commit after receiving $n-f-1$ \votetwo messages.
    
    Suppose that an honest party commits $v$ at Step~\ref{rb1:step:commit2}, then at least $n-f-1-(f-1)=n-2f\geq f+1$ honest non-broadcasters send \votetwo for $v$. We only need to show that no honest party send \votetwo for $v'\neq v$, then all honest parties will send \votetwo for $v$ and thus commit $v$. Suppose there is an honest party that send \votetwo for $v'\neq v$, then there exists two sets of \voteone messages from $n-f-1$ distinct non-broadcasters for $v$ and $v'$ respectively. Suppose there are $t>0$ Byzantine parties, then at least $n-f-1-(t-1)\geq n-t-f$ honest parties send \voteone for $v$ and $v'$ respectively, which is impossible as there are $n-t< 2(n-t-f)$ honest parties. Therefore no honest party sends \votetwo for $v'\neq v$, and all honest parties commits $v$.
    
    It is clear from the protocol that after at most 2 rounds (\voteone and \votetwo) since any honest party commits, all honest parties also commit. Hence the \badlatency is 4 rounds.
\end{proof}

\subsection{3-round Lower Bound for Unauthenticated BRB under $n\leq 4f-1$}
\label{sec:async:lowerbound}

\begin{figure*}[tb]
    \centering
    \includegraphics[width=0.9\textwidth]{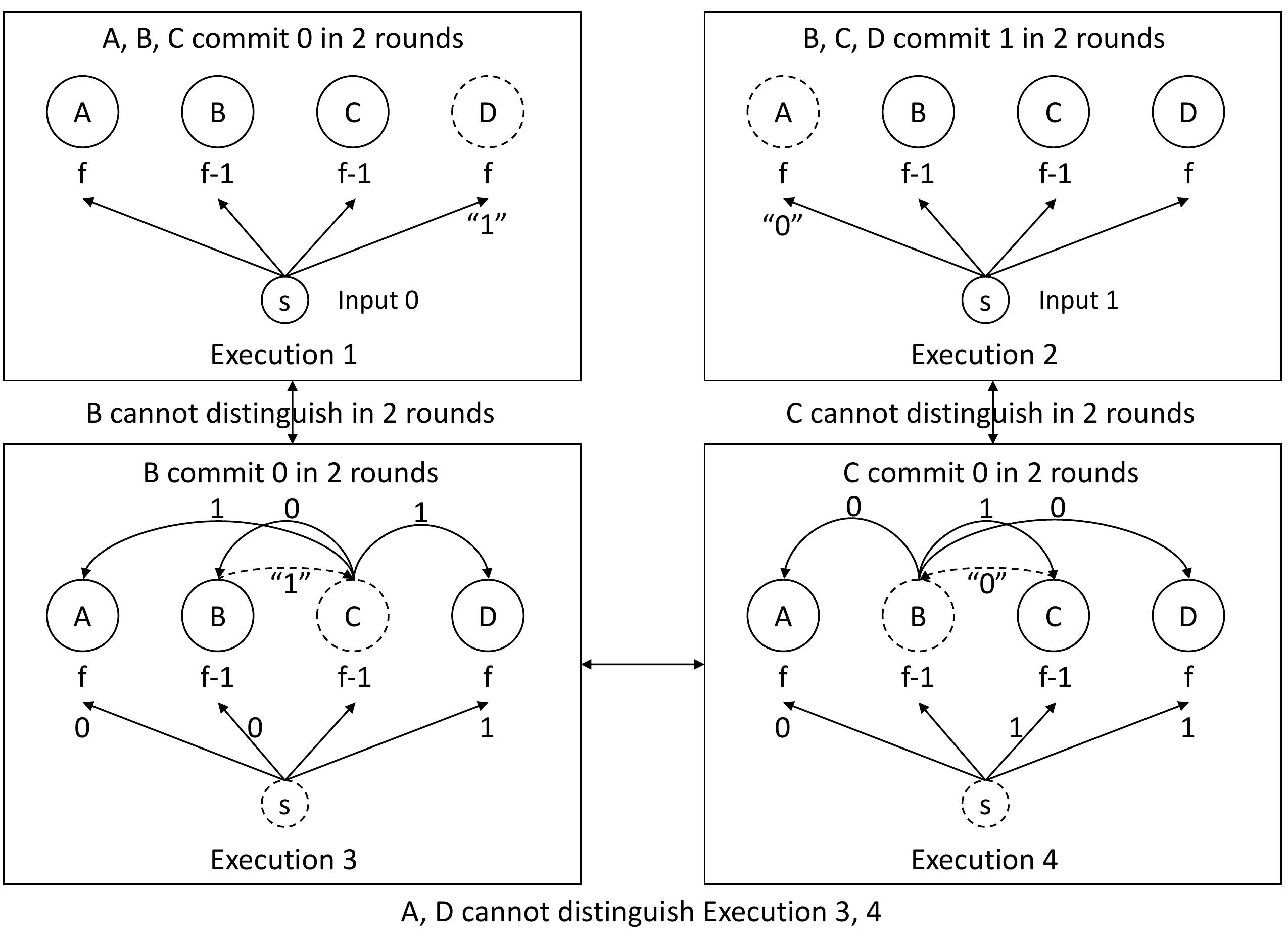}
    \caption{Unauthenticated BRB Good-case Latency Lower Bound: $3$ rounds under $n=4f-1$. 
    Dotted circles denote Byzantine parties. 
    }
    \label{fig:async:lb}
    % \vspace{-1em}
\end{figure*}

\begin{theorem}\label{thm:lowerbound}
    Any unauthenticated Byzantine reliable broadcast protocol under $3f+1\leq n \leq 4f-1$ must have a \latency of at least $3$ rounds even under synchrony.
\end{theorem}

\begin{proof}[Proof of Theorem~\ref{thm:lowerbound}.]
    The proof is illustrated in Figure~\ref{fig:async:lb}.
    We assume all parties start their protocol at the same time, which strengthens the lower bound result. Under synchrony, any message between all honest parties will be delivered within $\delta$ time, and hence each round of the protocol is of $\delta$ time. 
    Without loss of generality, we prove the lower bound for $n=4f-1$.
    Suppose there exists a BRB protocol $\Pi$ that has a \latency of $2$ round, which means the honest parties can always commit after receiving two rounds of messages but before receiving any message from the third round, if the designated broadcaster is honest.
    Let party $s$ be the broadcaster, and divide the remaining $n-1=4f-2$ parties into $4$ groups $A,B,C,D$ where $|A|=|D|=f$ and $|B|=|C|=f-1$. 
    For brevity, we often use $A$ ($B,C,D$) to refer all the parties in $A$ ($B,C,D$).
    Consider the following three executions of $\Pi$. 
    % All the executions constructed below have message delays equal to $\Delta$.
\begin{itemize}
    \item Execution  1. 
    The broadcaster $s$ is honest and has input $0$. Parties in $D$ are Byzantine, they behave honestly according to the protocol except that they pretend to receive from a broadcaster whose input is $1$. 
    Since the broadcaster is honest, by validity and \latency, parties in $A,B,C$ will commit $0$ after receiving two rounds of messages but before receiving any message from the third round.
    
    \item Execution  2. This execution is a symmetric case of Execution 1.
    The broadcaster $s$ is honest and has input $1$. Parties in $A$ are Byzantine, they behave honestly according to the protocol except that they pretend to receive from a broadcaster whose input is $0$. 
    Since the broadcaster is honest, by validity and \latency, parties in $B,C,D$ will commit $1$ after receiving two rounds of messages but before receiving any message from the third round.
    
    % \rl{I buy the intuition but I don't think this proof is rigorous. This seems to assume the protocol must work in a particular way, i.e., the broadcaster sends its input to everyone in the first round, and partys echo what the broadcaster sends }\daniel{fixed}
    
    \item Execution  3. 
    The broadcaster $s$ and the parties in $C$ are Byzantine. $s$ behaves to $A,B$ identically as in Execution 1 and to $D$ identically as in Execution 2. 
    Parties in $C$ behave to $B$ honestly according to the protocol except that they pretend to receive the same messages from the broadcaster as in Execution 1, and only send messages to $B$ in the first two rounds.
    Parties in $C$ behave to $A,D$ honestly except that they pretend to receive the same messages from the broadcaster as in Execution 2, and pretend to receive messages from $B$ as in Execution 2 only in the first two rounds.
    
    \item Execution  4. This execution is a symmetric case of Execution 3.
    The broadcaster $s$ and the parties in $B$ are Byzantine. $s$ behaves to $A$ identically as in Execution 1 and to $C,D$ identically as in Execution 2. 
    Parties in $B$ behave to $C$ honestly according to the protocol except that they pretend to receive the same messages from the broadcaster as in Execution 2, and only send messages to $C$ in the first two rounds.
    Parties in $B$ behave to $A,D$ honestly except that they pretend to receive the same messages from the broadcaster as in Execution 1, and pretend to receive messages from $C$ as in Execution 1 only in the first two rounds.
    
\end{itemize}
We show the following indistinguishability and contradiction.
\begin{itemize}
    \item $B$ cannot distinguish Execution 1 and 3 in the first two rounds, and thus will commit $0$ in the end of round 2 in Execution 3. The broadcaster $s$ behaves to $B$ identically in both executions. The messages sent to $B$ in the first round by any non-broadcaster party are identical in Execution 1 and 3, since the first round message only depends on the initial state and all Byzantine parties behave honestly in the first round.  For the second round, in Execution 3, since parties in $D$ pretends to $B$ that it receives messages from the broadcaster with input $1$, and parties in $C$ pretends to $B$ that it receives the same messages from the broadcaster as in Execution 1, the parties in $B$ observe the same messages in the first two rounds of both executions. Hence, $B$ cannot distinguish Execution 1 and 3 in the first two rounds. Since $B$ commit $0$ in the end of round 2 in Execution 1 due to validity and good-case latency, $B$ also commit $0$ in the end of round 2 in Execution 3.
    
    \item Similarly, $C$ cannot distinguish Execution 2 and 4 in the first two rounds, and thus will commit $1$ in the end of round 2 in Execution 4.
    
    \item $A,D$ cannot distinguish Execution 3 and 4. Similarly, the messages sent to $A,D$ in the first round are identical in both executions. 
    The broadcaster $s$ behaves to $A,B$ identically in Execution 3 and 4 as in Execution 1, and to $C,D$ identically in Execution 3 and 4 as in Execution 2. 
    In Execution 3, parties in $B$ only receive messages from $C$ in the first two rounds, and Byzantine parties in $C$ pretend to receive messages from a broadcaster whose input is $1$. In Execution 4, Byzantine parties in $B$ pretends only receiving two rounds of messages from $C$. Since the first two rounds of messages only depend on the initial state and the message received from the broadcaster in the first round, parties in $B$ receives the same messages from $C$. Therefore, $A,D$ receive the same messages from $B$ in both Execution 3 and 4.
    Similarly, $A,D$ receive the same messages from $C$ in both Execution 3 and 4, and thus cannot distinguish these two executions.
\end{itemize}
    
    {\em Contradiction.} 
    Since parties in $B$ commit $0$ in Execution 3, parties in $C$ commit $1$ in Execution 4, and parties in $A,D$ cannot distinguish Execution 3 and 4, either agreement or termination of BRB will be violated. Therefore no such protocol $\Pi$ exists.
\end{proof}

\section{Bad-case Latency of Asynchronous Byzantine Reliable Broadcast}
\label{sec:async:badcase}

%\rl{This is just worst-case, no need for the notion of bad-case}\daniel{I prefer bad-case, worst-case is no termination. see my definition of bad-case latency.}

In this section, we present 2 impossibility results and 4 asynchronous BRB protocols with tight trade-offs between resilience, \latency and \badlatency.

Recall that the classic Bracha reliable broadcast~\cite{bracha1987asynchronous} has optimal resilience of $n\geq 3f+1$, non-optimal \latency of 3 rounds and \badlatency of 4 rounds (1 extra round).
The 2-round BRB protocol by Imbs and Raynal~\cite{imbs2016trading} has non-optimal resilience of $n\geq 5f+1$, optimal \latency of 2 rounds and \badlatency of 3 rounds (1 extra round).
Meanwhile, our 2-round BRB protocol from Section~\ref{sec:async:goodcase} has optimal resilience $n\geq 4f$, optimal \latency of 2 rounds and \badlatency of 4 round (2 extra rounds).
All protocols above have optimal communication complexity of $O(n^2)$, matching the lower bound~\cite{dolev1985bounds}.

On the other hand, we can show that for any $f>1$, no asynchronous BRB protocol can achieve both \latency of 2 rounds and \badlatency of 2 rounds (Theorem~\ref{thm:simplelb} in Section~\ref{sec:simplelb}). For the special case of $f=1$, we show it is possible to have a $(2,2)$-round BRB (Theorem~\ref{thm:f=1}).

Therefore, it is interesting to ask:

\emph{
Under what conditions can BRB achieve optimality in all three metrics -- optimal resilience of $n\geq 4f$, optimal \latency of 2 rounds and optimal \badlatency of 3 rounds (1 extra round)?    
}

We show it is impossible for the general case of $f\geq 3$, by proving that no BRB protocol under $n\leq 5f-2,f\geq 3$ can achieve $(2,3)$-round (Theorem~\ref{thm:impossibility}). 
For $f\geq 3$, our BRB (Figure~\ref{fig:brb}) in earlier Section~\ref{sec:async:protocol} has optimal \latency of 2 rounds and optimal resilience $n\geq 4f$, but with \badlatency of 4 rounds.
On the other hand, we give a $(2,3)$-round BRB protocol (Figure~\ref{fig:5f-1}) with tight resilience $n\geq 5f-1$, improving the $n\geq 5f+1$ resilience of Imbs and Raynal~\cite{imbs2016trading}. 
For the special case of $f=2$, we show it is possible to construct a $(2,3)$-round BRB (Figure~\ref{fig:f=2}) with optimal resilience $n\geq 4f$.

% We make some progress in answering this question, by presenting several protocols that have optimal \latency of 2 rounds.
% For $f=1$, we have $(2,2)$-round BRB (Figure~\ref{fig:f=1}) with $n\geq 4f$, and for $f=2$ we have $(2,3)$-round BRB (Figure~\ref{fig:f=2}) with $n\geq 4f$.
% For $f\geq 3$, our BRB (Figure~\ref{fig:brb}) in earlier Section~\ref{sec:async:protocol} has optimal \latency of 2 rounds and optimal resilience $n\geq 4f$, but with \badlatency of 4 rounds.
% On the other hand, we give a $(2,3)$-round BRB protocol (Figure~\ref{fig:5f-1}) with resilience $n\geq 5f-1$, improving the $n\geq 5f+1$ resilience of Imbs and Raynal~\cite{imbs2016trading}. 
% \IA{maybe it is optimal? maybe say its better than Raynal etal?}

\subsection{Impossibility of $(2,2)$-round BRB}
\label{sec:simplelb}
For the general case of $f\geq 2$, we show any asynchronous BRB protocol cannot achieve $(2,2)$-round.

\begin{theorem}\label{thm:simplelb}
    Any asynchronous unauthenticated Byzantine reliable broadcast protocol under $f\geq 2$ and has a \latency of 2 rounds must have a \badlatency of at least 3 rounds.
\end{theorem}

\begin{proof}
    Suppose on the contrary there exists an asynchronous BRB protocol $\Pi$ that tolerates $f= 2$ and has $(2,2)$-round. We assume $n\geq 4f=8$, otherwise no protocol can solve BRB with \latency of 2 rounds by Theorem~\ref{thm:lowerbound}. Denote the broadcaster as party $0$ always, and remaining parties as party $1,...,n-1$.
    We construct the following executions.
    \begin{itemize}
        \item Execution 1. The broadcaster is honest, and has input $0$. Party $n-1$ is Byzantine and remain silent. Then all honest parties commit 0 in 2 rounds by assumption.
        \item Execution 2. 
        The broadcaster is Byzantine, and behaves honestly to parties $1,...,n-3$ with input $0$, and remains silent to other parties. Party $n-2$ is Byzantine, and behaves identically to party $n-3$ as in Execution 1, but remains silent to rest of the parties. Any messages from party $n-1$ are delayed and not delivered in 2 rounds.
        It is easy to see that party $n-3$ cannot distinguish Execution 2 and 1 in 2 rounds, therefore it commits 0 in 2 rounds in Execution 2 as well. By assumption, parties $1,...,n-4$ also commit 0 in 2 rounds in Execution 2.
        \item Execution $x$ for $x=3,...,n-3$. 
        The broadcaster is Byzantine, and behaves honestly to parties $1,...,n-x-1$ with input $0$, and remains silent to other parties. Party $n-x$ is Byzantine, and behaves identically to party $n-x-1$ as in Execution $x-1$, but remains silent to rest of the parties. Any messages from party $n-x+1$ are delayed and not delivered in 2 rounds.
        It is easy to see that party $n-x-1$ cannot distinguish Execution $x$ and $x-1$ in 2 rounds, therefore it commits 0 in 2 rounds in Execution $x$ as well. By assumption, parties $1,...,n-x-2$ also commit 0 in 2 rounds in Execution 2.
        \item Execution $n-2$. 
        The broadcaster is Byzantine, and behaves honestly to party $1$ with input $0$, and remains silent to other parties. Party $2$ is Byzantine, and behaves identically to party $1$ as in Execution $n-3$, but remains silent to rest of the parties. Any messages from party $3$ are delayed and not delivered in 2 rounds.
        It is easy to see that party $1$ cannot distinguish Execution $n-2$ and $n-3$ in 2 rounds, therefore it commits 0 in 2 rounds in Execution $n-2$ as well.
    \end{itemize}
    Similarly, we can construct $n-2$ symmetric executions, where the broadcaster has input 1, and in the last execution the broadcaster only behaves honestly to party $n-1$ with input 1, and party $n-1$ commits 1 in 2 rounds. 
    
    {\em Contradiction.} Now we consider another execution, where the broadcaster is Byzantine, it behaves to party 1 honestly with input 0, and to party $n-1$ honestly with input 1, and remain silent to other parties. Party 2 is Byzantine, it behaves to party 1 identically as in Execution $n-3$, and to party $n-1$ identically as the party $n-2$ to party $n-1$ in the last execution of the constructed symmetric executions (due to symmetric of the non-broadcaster parties, the index does not matter). Any messages between parties $1,n-1$ are delayed and not delivered in 2 rounds. Then, party 1 commits 0 in 2 rounds while party $n-1$ commit 1 in 2 rounds, breaking agreement of the BRB.
    Therefore, such protocol $\Pi$ does not exist.
\end{proof}

\subsection{$(2,2)$-round BRB Protocol under $n\geq 4f, f=1$}
For the special case of $f=1$, we can show a simple BRB protocol (Figure~\ref{fig:f=1}) that has optimal \latency and \badlatency of 2 rounds, while having optimal resilience $n\geq 4$.

\begin{figure}[tb]
    \centering
    \begin{mybox}
\begin{enumerate}%[itemsep=0pt,topsep=0pt]
    \item\label{rb4:step:propose} \textbf{Propose.} The designated broadcaster $L$  with input $v$ sends $\langle \texttt{propose}, v\rangle$ to all  parties.

    \item\label{rb4:step:ack} \textbf{Ack.} 
    When receiving the first proposal $\langle \texttt{propose}, v\rangle$ from the broadcaster,
    a party sends an \ack message for $v$ to all parties in the form of $\langle \ack, v \rangle$.

    \item\label{rb4:step:commit1} \textbf{2-round Commit.}
    When receiving $\langle \ack, v \rangle$ from $n-2$ distinct non-broadcaster parties, a party commits $v$ and terminates.

\end{enumerate}
    \end{mybox}
    % \vspace{-1em}
    \caption{$(2,2)$-round BRB Protocol under $n\geq 4f, f=1$}
    % \vspace{-1em}
    \label{fig:f=1}
\end{figure}

\begin{theorem}\label{thm:f=1}
    The protocol in Figure~\ref{fig:f=1} solves Byzantine reliable broadcast under asynchrony with optimal resilience $n\geq 4, f=1$, optimal \latency and \badlatency of 2 rounds.
\end{theorem}

\begin{proof}
    {\bf Validity and Good-case Latency.}
    
    If the broadcaster is honest, it sends the same proposal of value $v$ to all parties, and all $n-2\geq 2$ non-broadcaster honest parties will multicast the \ack message for $v$. Since there is just one Byzantine party, its \ack is below the $n-2$ threshold. Then all honest parties will commit $v$ after receiving $n-2$ \ack messages at Step~\ref{rb4:step:commit1} and terminate. The \latency is $2$ rounds, including broadcaster sending the proposal and all parties sending \ack message.
    
    {\bf Agreement, Termination and Bad-case Latency.}
    
    If the broadcaster is honest, by validity all honest parties will commit the same value. 
    If the broadcaster is Byzantine, then all $n-1$ non-broadcaster parties are honest. If an honest party commits $v$ at Step~\ref{rb4:step:commit1}, then it receives $n-2$ \ack messages of $v$ from distinct non-broadcaster parties, and thus all honest parties will also receive these \ack messages and commit $v$.
    Since all honest parties commit in the same asynchronous round, the \badlatency is also 2 rounds.
\end{proof}

\subsection{Impossibility of $(2,3)$-round BRB}

\begin{theorem}\label{thm:impossibility}
    Any asynchronous unauthenticated Byzantine reliable broadcast protocol under $n\leq 5f-2, f\geq 3$ and has a \latency of 2 rounds must have a \badlatency of at least 4 rounds.
\end{theorem}

\begin{proof}

    Suppose on the contrary that there exists an asynchronous BRB protocol $\Pi$ under $n=5f-2, f\geq 3$ that has $(2,3)$-round. Denote the broadcaster as party 0 always, denote 2 non-broadcaster parties as $p,q$, and divide the remaining $5f-5$ parties into 5 groups $G_1,...,G_5$ each of size $f-1$ (recall $f-1\geq 2$). 
    Denote $G_L = \{p\}\cup G_1\cup G_2$ and $G_R=G_4\cup G_5\cup \{q\}$.
    % Denote $\mathrm{G}=\bigcup_{1\leq i\leq 5}G_i$, and $G_{-i}=\mathrm{G}\setminus G_i$.
    We use $S[i]$ to denote the $i$-th party in set $S$, where $S$ can be any set defined above (such as $G_{j}$ for $j=1,...,5$ and $G_L, G_R$).
    We construct the following executions.
    In all constructed executions, all messages are delivered by the recipient after $\Delta$ time by default, and we will explicitly specify the messages that are delayed by the adversary due to asynchrony.
    
    \begin{itemize}
        \item $E_1^0$.
        The broadcaster is honest and has input 0. Parties in $G_5\cup \{q\}$ are Byzantine, and they behave honestly except that they pretend to receive from a broadcaster whose input is 1. 
        Since the broadcaster is honest, by validity and \latency, all honest parties commit 0 after receiving two rounds of messages.
        
        \item $E_1^1$. This execution is a symmetric case of $E_1^0$.
        The broadcaster is honest and has input 1. Parties in $G_1\cup \{p\}$ are Byzantine, and they behave honestly except that they pretend to receive from a broadcaster whose input is 0. 
        Since the broadcaster is honest, by validity and \latency, all honest parties commit 1 after receiving two rounds of messages.
        
        \item $E_2^0$. 
        The broadcaster is Byzantine, it behaves to $G_L\cup G_3$ identically as in $E_1^0$, and to $G_5\cup \{q\}$ identically as in $E_1^1$. Parties in $G_4$ are Byzantine, they behave to the party $G_3[f-1]$ honestly but pretending to receive from the broadcaster in $E_1^0$, and to other parties honestly but pretending to receive from the broadcaster in $E_1^1$.
        
        \textbf{Claim:} 
        The honest party $G_3[f-1]$ cannot distinguish $E_2^0$ and $E_1^0$ in 2 rounds, and thus will commit 0 in round 2. 
        Then, by assumption, all honest parties also commit 0 in round 3 in $E_2^0$.
        The broadcaster behaves to $G_3[f-1]$ identically in both executions. The messages sent to $G_3[f-1]$ in the first round by any non-broadcaster party are identical in $E_2^0$ and $E_1^0$, since the first round message only depends on the initial state and all Byzantine parties behave honestly in the first round. For the second round, since in $E_1^0$ the Byzantine parties in $G_5\cup \{q\}$ pretend to receive from a broadcaster with input 1, they send the same round-2 messages as in $E_2^0$. For the Byzantine parties in $G_4$, they behave identically to $G_3[f-1]$ by construction. All honest parties in $G_L$ also behave identically to $G_3[f-1]$ in round 2 since they receive the same round-1 messages.
        Therefore party $G_3[f-1]$ cannot distinguish $E_2^0$ and $E_1^0$ in 2 rounds, and thus will commit 0 in round 2.
        
        \item $E_3^0$.
        The broadcaster is Byzantine, it behaves to $G_L$ identically as in $E_1^0$, and to $G_R$ identically as in $E_1^1$. Parties in $G_3$ are Byzantine, they behave to other parties identically as in $E_2^0$.
        
        \textbf{Claim:}
        The honest parties in $G_L\cup G_R$ cannot distinguish $E_3^0$ and $E_2^0$ in 3 rounds, and thus will commit 0 in round 3. For the round-1 message, honest parties receive the same messages in both executions since Byzantine parties including the broadcaster send the same messages. For the round-2 message, the Byzantine parties of $G_4$ in $E_2^0$ behave to $G_L\cup G_R$ as if they receive from a broadcaster with input 1, which would be identically to $E_3^0$. The Byzantine parties of $G_3$ in $E_3^0$ behave to other parties identically as in $E_2^0$ by construction. Similarly, $G_l\cup G_R$ receive the same round-3 messages in both executions, and thus cannot distinguish $E_3^0$ and $E_2^0$ in 3 rounds, and will commit 0 in round 3 in $E_3^0$ as well.
        
        \item $E_{2j+2}^0$ for $j=1,2,...,|G_R|=2f-1$. 
        The broadcaster is Byzantine, it behaves to $G_L\cup \{G_3[1]\}$ identically as in $E_1^0$, and to $G_R$ identically as in $E_1^1$. Parties in $G_3\setminus \{G_3[1]\}$ are Byzantine, and they behave to all honest parties identically as in $E_{2j+1}^0$. Party $G_R[j]$ is Byzantine, and it behaves to all honest parties except $p$ identically as in $E_{2j+1}^0$, and to party $p$ honestly except that it pretends receiving no message from $G_3$ sent after round 1.
        
        \item $E_{2j+3}^0$ for $j=1,2,...,|G_R|=2f-1$.
        The broadcaster is Byzantine, it behaves to $G_L$ identically as in $E_1^0$, and to $G_R$ identically as in $E_1^1$. Parties in $G_3$ are Byzantine, they behave to other parties identically as in $E_{2j+2}^0$, but they send no message to $G_R[j]$ after round 1.

        \textbf{Claim:}
        Any honest party in $G_L\setminus \{p\}$ cannot distinguish $E_{2j+2}^0$ and $E_{2j+1}^0$ in 3 rounds, and it will commit 0 in round 3 in $E_{2j+2}^0$.
        Then, by assumption, party $p$ will also commit 0 in round 3 in $E_{2j+2}^0$. 
        Similar to previous claim, honest parties receive the same round-1 messages. For round 2, the Byzantine parties in $E_{2j+2}^0$ behave identically to all honest parties, including party $p$ since the difference from $G_R[j]$ to $p$ is reflected only after round 2. Hence, honest parties in $G_L\setminus \{p\}$ will also receive the same messages in round 3, thus cannot distinguish $E_{2j+2}^0$ and $E_{2j+1}^0$ in 3 rounds.
        
        \textbf{Claim:}
        Party $p$ cannot distinguish $E_{2j+3}^0$ and $E_{2j+2}^0$ in 3 rounds, and thus will commit 0 in round 3 in $E_{2j+3}^0$. 
        Then, by assumption, all honest parties in $G_L\cup G_R$ also commit 0 in round 3.
        Similar to previous claim, honest parties receive the same round-1 and round-2 messages.
        For round 3, since Byzantine parties in $G_3$ send no message to $G_R[j]$ after round 1 in $E_{2j+3}^0$, the honest party $G_R[j]$ in $E_{2j+3}^0$ will behave the same to $p$ as the Byzantine party $G_R[j]$ which pretends to $p$ that it receives no message from $G_3$ in $E_{2j+2}^0$. 
        Hence, $p$ cannot distinguish $E_{2j+3}^0$ and $E_{2j+2}^0$ in 3 rounds.
        
        % \item $E_{2j+4}^0$ for $j=1,2,...,4(f-1)$.
        % These series of executions are similar to $E_{2j+3}^0$ but we ``swap'' parties $G_3[1]$ and $G_{-3}[j]$. 
        % The broadcaster is Byzantine, it behaves to $\{p\}\cup G_1\cup G_2\cup G_3[1]$ identically as in $E_1^0$, and to $G_4\cup G_5\cup \{q\}$ identically as in $E_1^1$.
        % Parties in $G_3\setminus \{G_3[1]\}$ are Byzantine, and they behave to all honest parties identically as in $E_{2j+3}^0$. Party $G_{-3}[j]$ is Byzantine, and it behaves to all honest parties except $p$ identically as in $E_{2j+3}^0$, and to party $p$ honestly except that it pretends receiving no message from $G_3$.
        % Party $q$ cannot distinguish $E_{2j+4}^0$ and $E_{2j+3}^0$ in 3 rounds.
        % \item $E_{2j+5}^0$ for $j=1,2,...,4(f-1)$.
        % The broadcaster is Byzantine, it behaves to $\{p\}\cup G_1\cup G_2$ identically as in $E_1^0$, and to $G_4\cup G_5\cup \{q\}$ identically as in $E_1^1$. Parties in $G_3$ are Byzantine, they behave to other parties identically as in $E_{2j+4}^0$, but they send no message to $G_{-3}[j]$.
        % Party $p$ cannot distinguish $E_{2j+5}^0$ and $E_{2j+4}^0$ in 3 rounds.
    \end{itemize}
    By the above constructions, we finally have an execution $E_{2j+3,j=2f-1}^0=E_{4f+1}^0$ where the Byzantine broadcaster behaves to $G_L$ with input 0, and to $G_R$ with input 1, and the Byzantine parties in $G_3$ send no message to $G_R$, but party $p$ has to commit 0 in 3 rounds. 
    Similarly, we can construct a series of symmetric executions of the above executions including $E_1^1$, i.e., $E_1^1,E_2^1,...,E_{4f+1}^1$, and have the execution $E_{4f+1}^1$ where the Byzantine broadcaster also behaves to $G_L$ with input 0, and to $G_R$ with input 1, and the Byzantine parties in $G_3$ send no message to $G_L$, but party $q$ has to commit 1 in 3 rounds.
    
    \emph{Contradiction.}
    Now we construct another middle execution $E_m$, where
    the Byzantine broadcaster behaves to $G_L$ with input 0, and to $G_R$ with input 1, and Byzantine parties in $G_3$ behave to $G_L$ identically as in $E_{4f+1}^0$ and to $G_R$ identically as in $E_{4f+1}^1$.
    It is easy to see that party $p$ cannot distinguish $E_m$ and $E_{4f+1}^0$ in 3 rounds, and thus will commit 0 in round 3, while party $q$ cannot distinguish $E_m$ and $E_{4f+1}^1$ in 3 rounds, and thus will commit 1 in round 3. 
    This violates the agreement property of BRB, and hence such BRB protocol $\Pi$ does not exist.
\end{proof}

\subsection{$(2,3)$-round BRB Protocol under $n\geq 4f, f=2$}
For the special case of $f=2$, we propose a $(2,3)$-round BRB protocol (Figure~\ref{fig:f=2}) that has optimal resilience $n\geq 4f$.
The main idea is that all parties send \ack for broadcaster's proposal, and also send \vote for other parties' \ack. When receiving enough \vote messages of $v$ for the same party, a party locks on $v$. The protocol guarantees that all honest parties lock on the same value for each party when $f=2$. Then, the 3-round commit step let a party commits if the party locks on the same value for a majority of the parties. 
Since all parties send a \vote for all other parties, the message and communication complexity are both $O(n^3)$.

\begin{figure}[tb]
    \centering
    \begin{mybox}
\begin{enumerate}%[itemsep=0pt,topsep=0pt]
    \item\label{rb3:step:propose} \textbf{Propose.} The designated broadcaster $L$  with input $v$ sends $\langle \texttt{propose}, v\rangle$ to all  parties.

    \item\label{rb3:step:ack} \textbf{Ack.} 
    When receiving the first proposal $\langle \texttt{propose}, v\rangle$ from the broadcaster,
    a party sends a \ack message for $v$ to all parties in the form of $\langle \ack, v \rangle$.
    
    \item\label{rb3:step:commit1} \textbf{2-round Commit.}
    When receiving $\langle \ack, v \rangle$ from $n-f-1$ distinct non-broadcaster parties, a party commits $v$ and terminates.
    
    \item\label{rb3:step:vote} \textbf{Vote and Lock.} 
    \begin{itemize}
        \item When receiving $\langle \ack, v \rangle$ from a non-broadcaster party $j$, a party sends $\langle \vote, j, v \rangle$ to all parties if not yet sent \vote for party $j$.
        \item When receiving $\langle \vote, j, v \rangle$ from $n-f-2$ distinct {\em non-broadcaster parties other than $j$}, a party locks on $v$ for party $j$.
    \end{itemize}

    \item\label{rb3:step:commit2} \textbf{3-round Commit.}
    When locking on the same $v$ for $n-2f$ distinct non-broadcaster parties, a party commits $v$ and terminates.

\end{enumerate}
    \end{mybox}
    % \vspace{-1em}
    \caption{$(2,3)$-round BRB under $n\geq 4f, f=2$}
    % \vspace{-1em}
    \label{fig:f=2}
\end{figure}

\begin{lemma}\label{lem:f=2}
    If the broadcaster is Byzantine and an honest party locks on $v$ for party $j$, then all honest parties also lock on $v$ for party $j$.
\end{lemma}

\begin{proof}
    Since an honest party locks on $v$ for party $j$, it receives $n-f-2$ \vote messages from non-broadcaster parties other than $j$. If $j$ is honest, then it sends the same \ack to all parties, and thus all honest parties receive $n-f-2$ \vote for party $j$ from non-broadcaster honest parties other than $j$. If $j$ is Byzantine, then the parties other than $j$ and the broadcaster are all honest. Since an honest party receives $n-f-2$ \vote messages from these honest parties, all honest parties will also receive the messages. Therefore, all honest parties also lock on $v$ for party $j$.
\end{proof}

\begin{theorem}\label{thm:f=2}
    The protocol in Figure~\ref{fig:f=2} solves Byzantine reliable broadcast under asynchrony with optimal resilience $n\geq 4f, f=2$ and optimal \latency of 2 rounds, and has \badlatency of 3 rounds.
\end{theorem}

\begin{proof}
    {\bf Validity and Good-case Latency.}
    
    If the broadcaster is honest, it sends the same proposal of value $v$ to all parties, and all $n-f-1$ honest non-broadcaster parties will multicast the \ack message for $v$. Since there are only $f$ Byzantine party, their \ack messages is below the $n-f-1$ threshold. Then all honest parties will commit $v$ after receiving $n-f-1$ \ack messages at Step~\ref{rb3:step:commit1} and terminate. The \latency is $2$ rounds, including broadcaster sending the proposal and all parties sending \ack message.
    
    {\bf Agreement.}
    
    If the broadcaster is honest, by validity all honest parties will commit the same value. Now consider when the broadcaster is Byzantine, and suppose there are $t>0$ Byzantine parties there are at most $t-1$ Byzantine parties among non-broadcasters.
    
    If any two honest parties commit different values at Step~\ref{rb3:step:commit1}, then there must be at least $n-f-1-(t-1)=n-f-t$ honest parties sending \ack for each of these different values. It is impossible by quorum intersection since there are only $n-t$ honest parties.
    
    Suppose any two honest parties commit different values at Step~\ref{rb3:step:commit2}. Then, there must exists at least $2(n-2f)-(n-1)\geq 1$ party for which the two committed honest parties lock different values. However, this contradicts Lemma~\ref{lem:f=2}, which states honest parties lock on the same value for any party when the broadcaster is Byzantine. 
    Hence, no two honest parties can commit different values at Step~\ref{rb3:step:commit2}.
    
    Now we show that if an honest party $h1$ commits $v$ at Step~\ref{rb3:step:commit1} and another honest party $h2$ commits $v'$ at Step~\ref{rb3:step:commit2}, then it must be $v=v'$. 
    Suppose $h1$ commits $v$ at Step~\ref{rb1:step:commit1}, then at least $n-f-1-(f-1)=n-2f$ honest non-broadcaster parties send \ack for $v$. All honest parties will lock on $v$ for these $n-2f$ non-broadcaster parties, which is a majority of the $n-1$ non-broadcaster parties. Therefore any honest party that commits $v'$ at Step~\ref{rb3:step:commit2} must have $v'=v$.
    
    {\bf Termination and Bad-case Latency.}
    If the broadcaster is honest, by validity all honest parties will commit the same value. 
    If the broadcaster is Byzantine, once an honest party commits $v$ at Step~\ref{rb3:step:commit1}, there are $n-2f$ non-broadcaster honest parties that send \ack for $v$, and all honest parties will eventually lock on $v$ for these parties after receiving the \vote messages. Therefore all honest parties will commit $v$ at Step~\ref{rb3:step:commit2} after 1 extra round.
\end{proof}

\subsection{$(2,3)$-round BRB under $n\geq 5f-1$}
In this section, we improve the resilience of 2-round BRB protocol in the previous work~\cite{imbs2016trading} from $5f+1$ to $5f-1$, while keeping the \badlatency 3 rounds.
The protocol is presented in Figure~\ref{fig:5f-1}, and the main difference compared to Imbs and Raynal~\cite{imbs2016trading} is that in Step~\ref{rb2:step:ack}, parties send \ack for $v$ if receiving $n-2f$ \ack from {\em non-broadcaster parties}, instead of from any parties as in~\cite{imbs2016trading}. 
The intuition is that when the broadcaster is Byzantine, the above set of non-broadcaster parties only contains $f-1$ Byzantine parties, and thus we can reduce the total number of parties but still ensure quorum intersection.
% \IA{also the modification that we ignore the broadcaster!}\daniel{for the commit step, including the broadcaster or excluding it should be the same. n-f-1 non-broadcasters or n-f any parties.}\IA{are you sure? maybe you see n-f-2 acks but I see n-f-2 and the broadcasters's ack?}

\begin{figure}[tb]
    \centering
    \begin{mybox}
\begin{enumerate}%[itemsep=0pt,topsep=0pt]
    \item\label{rb2:step:propose} \textbf{Propose.} The designated broadcaster $L$  with input $v$ sends $\langle \texttt{propose}, v\rangle$ to all  parties.

    \item\label{rb2:step:ack} \textbf{Ack.} 
    \begin{itemize}
    \item When receiving the first proposal $\langle \texttt{propose}, v\rangle$ from the broadcaster,
    a party sends a \ack message for $v$ to all parties in the form of $\langle \ack, v \rangle$.
    
    \item When receiving $\langle \ack, v \rangle$ from $n-2f$ distinct {\em non-broadcaster parties}, a party sends $\langle \ack, v \rangle$ to all parties if not yet sent $\langle \ack, v \rangle$.
    \end{itemize}

    \item\label{rb2:step:commit} \textbf{Commit.}
    When receiving $\langle \ack, v \rangle$ from $n-f-1$ distinct non-broadcaster parties, a party commits $v$ and terminates.

\end{enumerate}
    \end{mybox}
    % \vspace{-1em}
    \caption{$(2,3)$-round BRB under $n\geq 5f-1$}
    % \vspace{-1em}
    \label{fig:5f-1}
\end{figure}

\begin{theorem}\label{thm:5f-1}
    The protocol in Figure~\ref{fig:5f-1} solves Byzantine reliable broadcast under asynchrony with resilience $n\geq 5f-1$ and optimal \latency of 2 rounds, and has \badlatency of 3 rounds.
\end{theorem}

\begin{proof}
    {\bf Validity and Good-case Latency.}
    
    If the broadcaster is honest, it sends the same proposal of value $v$ to all parties, and all $n-f-1$ honest non-broadcaster parties will multicast the \ack message for $v$. Since there are only $f$ Byzantine party, their \ack messages is below the $n-f-1$ threshold. Then all honest parties will commit $v$ after receiving $n-f-1$ \ack messages at Step~\ref{rb1:step:commit1} and terminate. The \latency is $2$ rounds, including broadcaster sending the proposal and all parties sending \ack message.
    
    {\bf Agreement.}
    
    If the broadcaster is honest, by validity all honest parties will commit the same value. 
    If the broadcaster is Byzantine, and suppose there are $t>0$ Byzantine parties, then there are $t-1$ Byzantine parties among all non-broadcaster parties.
    Suppose that two honest parties commit different values $v\neq v'$, then by Step~\ref{rb2:step:commit} there are at least $n-f-1-(t-1)=n-f-t$ honest parties $A$ that send \ack for $v$ and at least $n-f-1-(t-1)=n-f-t$ honest parties $B$ that send \ack for $v'$. Since there are $n-t$ honest parties in total, $|A\cap B|\geq 2(n-f-t)-(n-t)=n-2f-t\geq 3f-t-1>0$, there must exist some honest party that sends \ack due to the second condition of Step~\ref{rb2:step:ack}. 
    If the above only happens to $v$, then there are at least $n-2f-(t-1)=n-2f-t+1$ honest parties that send \ack for $v$ due to receiving the \propose from the broadcaster. This contradicts the fact that at least $n-f-t$ honest parties send \ack for $v'$ due to receiving \propose, since $(n-2f-t+1)+(n-f-t)>n-t$. 
    It the above happens to both $v,v'$, then there are at least $n-2f-(t-1)=n-2f-t+1$ honest parties that send \ack for $v$ (and for $v'$, respectively) due to receiving the \propose from the broadcaster. This is also impossible since $2(n-2f-t+1)\geq n+f-2t+1 > n-t$. Therefore, all honest parties commit the same value.
    
    {\bf Termination and Bad-case Latency.}
    If the broadcaster is honest, by validity all honest parties will commit the same value. 
    If the broadcaster is Byzantine, once an honest party commits $v$, there are $n-2f$ non-broadcaster honest parties that send \ack for $v$. Therefore all honest parties will send \ack for $v$ and hence commit $v$ after 1 extra round.
\end{proof}

\section{Extension to Unauthenticated Byzantine Broadcast under Synchrony}
\label{sec:sync}
In this section, we extend the previous results to show the \latency results for unauthenticated Byzantine broadcast under synchrony.
It is well-known that unauthenticated Byzantine broadcast or Byzantine reliable broadcast is solvable if and only if $n\geq 3f+1$.

We adopt the synchrony model assumptions from~\cite{abraham2021goodcase}, including distinguishing the latency bounds $\delta$ and $\Delta$, and the clock assumption, briefly as follows. More details about the model assumptions can be found in~\cite{abraham2021goodcase}.

\paragraph{Network delays.}
We separate the {\em actual bound $\delta$}, and the {\em conservative bound $\Delta$} on the network delay:
\begin{itemize}[itemsep=0pt,topsep=0pt]
    \item For one execution, $\delta$ is the upper bound for message delays between any pair of honest parties, but the value of $\delta$ is {\em unknown} to the protocol designer or any party. Different executions may have different $\delta$ values. 
    \item For all executions, $\Delta$ is the upper bound for message delays between any pair of honest parties, and the value of $\Delta$ is {\em known} to the protocol designer and all parties. 
\end{itemize}

\paragraph{Clock synchronization.}
Each party is equipped with a local clock that starts counting at the beginning of the protocol execution. 
We assume the \emph{clock skew} is at most $\clockskew$, i.e., they start the protocol at most $\clockskew$ apart from each other.
We assume parties have {\em no clock drift} for convenience. 
There exist clock synchronization protocols~\cite{dolev1995dynamic, abraham2019synchronous} that guarantee a bounded clock skew of $\clockskew\leq \delta$.
Since the value of $\delta$ is unknown to the protocol designer or any party, our protocol will use $\Delta$ as the parameter for clock skew in the protocol. Note that the actual clock skew is still $\clockskew\leq \delta$, guaranteed by the clock synchronization protocols~\cite{dolev1995dynamic, abraham2019synchronous}.
In addition, due to clock skew, the BA primitive used in our BB protocol (Figure~\ref{fig:bb}) needs to tolerate up to $\clockskew$ clock skew.
For instance, any synchronous lock-step BA can do so by using a clock synchronization algorithm \cite{dolev1995dynamic, abraham2019synchronous} to ensure at most $\Delta$ clock skew, and setting each round duration to be $2\Delta$ to enforce the abstraction of lock-step rounds.

\begin{theorem}\label{thm:sync}
    Any unauthenticated Byzantine reliable broadcast protocol under $3f+1\leq n \leq 4f-1$ must have a \latency of at least $3\delta$ under synchrony.
\end{theorem}

The proof of Theorem~\ref{thm:sync} is analogous to that of Theorem~\ref{thm:lowerbound}, and is omitted here for brevity.
Next, we show a synchronous BB protocol in Figure~\ref{fig:bb} that has \latency of $2\delta$ under $n\geq 4f$.

% \rl{Tight in what sense? You need to invoke some result that says 2delta is a lower bound for all n and f?}

\paragraph{Protocol description.}
The protocol is presented in Figure~\ref{fig:bb}, and is inspired by our $(2,4)$-round asynchronous BRB protocol (Figure~\ref{fig:brb}) from Section~\ref{sec:async:protocol}. The main idea is to add a Byzantine agreement at the end of the protocol to ensure termination, since BRB does not require termination when the broadcaster is Byzantine. The input of the BA is called \lock, which is set to be some default value $\bot$ initially, and will be set when commit in Step~\ref{bb:step:commit} or receiving enough \vote in Step~\ref{bb:step:vote}. One guarantee implied by the $(2,4)$-round BRB protocol is that, when any honest party commit $v$ in Step~\ref{bb:step:commit}, all honest parties will lock on $v$, and therefore the BA will only output $v$.

\begin{figure}[tb]
    \centering
    \begin{mybox}
Initially, every party $i$ starts the protocol at most $\delta$ time apart with a local clock and sets $\lock=\bot$, $\clockskew=\Delta$.

\begin{enumerate}%[itemsep=0pt,topsep=0pt]
    \item\label{bb:step:propose} \textbf{Propose.} The designated broadcaster $L$  with input $v$ sends $\langle \texttt{propose}, v\rangle$ to all  parties.

    \item\label{bb:step:ack} \textbf{Ack.} 
    When receiving the first proposal $\langle \texttt{propose}, v\rangle$ from the broadcaster,
    a party sends an \ack message for $v$ to all parties in the form of $\langle \ack, v \rangle$.

    \item\label{bb:step:commit} \textbf{Commit.}
    When receiving $\langle \ack, v \rangle$ from $n-f-1$ distinct non-broadcaster parties at time $t$, a party sets $\lock=v$. If $t\leq 2\Delta+\clockskew$, the party commits $v$.
    
    \item\label{bb:step:vote} \textbf{Vote.} 
    \begin{itemize}
        \item When receiving $\langle \ack, v \rangle$ from $n-2f$ distinct non-broadcaster parties, a party sends a \vote message for $v$ to all parties in the form of $\langle \vote, v \rangle$ if not yet sent \vote.
        \item When receiving $\langle \vote, v \rangle$ from $n-f-1$ distinct non-broadcaster parties, a party sets $\lock = v$.
        % \item When receiving $\langle \votetwo, v \rangle$ from $f+1$ distinct non-broadcaster parties, send a \votetwo message for $v$ to all parties in the form of $\langle \votetwo, v \rangle$ if not yet sent \votetwo and set $\lock = v$.
    \end{itemize}

    \item\label{bb:step:ba} \textbf{Byzantine agreement.}
    At local time $3\Delta+2\clockskew$, a party invokes an instance of Byzantine agreement with \lock as the input. If not committed, the party commits on the output of the Byzantine agreement. Terminate.

\end{enumerate}
    \end{mybox}
    % \vspace{-1em}
    \caption{$2\delta$ unauthenticated BB protocol under $n\geq 4f$}
    % \vspace{-1em}
    \label{fig:bb}
\end{figure}

\begin{theorem}\label{thm:ub:sync}
    The protocol in Figure~\ref{fig:bb} solves Byzantine broadcast under synchrony with optimal resilience $n\geq 4f$ and optimal \latency of $2\delta$.
\end{theorem}

\begin{proof}
    {\bf Validity and Good-case Latency.}
    If the broadcaster is honest, it proposes the same value $v$ to all parties, and all honest parties will send \ack for $v$.
    Then at Step~\ref{bb:step:commit}, all honest parties receive $n-f-1$ \ack messages of $v$ after $2\delta$ time (which is before local time $2\Delta+\clockskew$), and commits $v$.
    
    {\bf Agreement.}
    If all honest parties commit at Step~\ref{bb:step:ba}, all honest parties commit on the same value due to the agreement property of the BA.
    Otherwise, there must be some honest party that commits at Step~\ref{bb:step:commit}.
    First, no two honest parties can commit different values at Step~\ref{bb:step:commit} due to quorum intersection.
    Now suppose any honest party $h$ that commits $v$ at Step~\ref{bb:step:commit}. If the broadcaster is honest, by validity, all honest parties commits $v$. If the broadcaster is Byzantine, then there are $f-1$ Byzantine parties among non-broadcasters.
    Since $h$ receives $n-f-1$ \ack messages from non-broadcasters, at least $n-f-1-(f-1)=n-2f$ of them are from honest parties. Then, all honest parties receive these $n-2f$ \ack messages and set $\lock =v$ at their local time $\leq (2\Delta+\clockskew)+\Delta+\clockskew= 3\Delta+2\clockskew$, before invoking the Byzantine agreement primitive at Step~\ref{bb:step:ba}, since the clock skew is $\clockskew$ and message delay is bounded by $\Delta$. Also by quorum intersection, there cannot be $n-2f$ \ack messages for $v'\neq v$, since the set of $(n-2f)-(f-1)=n-3f+1$ honest parties who voted for $v'$ and the set of $n-2f$ honest parties who voted for $v$ intersect at $\geq (n-3f+1)+(n-2f)-(n-f)\geq 1$ honest parties.
    Therefore, at Step~\ref{bb:step:ba}, all honest parties have the same input $\lock=v$ to the BA. Then by the validity condition of the BA primitive, the output of the agreement is also $v$. Any honest party that does not commit at Step~\ref{bb:step:commit} will commit $v$ at Step~\ref{bb:step:ba}. 
    
    {\bf Termination.}
    According to the protocol, honest parties terminate at Step~\ref{bb:step:ba}, and they commit a value before termination.
    
\end{proof}

\section{Related Work}

% \paragraph{Byzantine fault-tolerant broadcast protocols.}
Byzantine fault-tolerant broadcast, first proposed by Lamport et al. \cite{lamport1982byzantine}, have received a significant amount of attention for several decades. 
Under synchrony, the deterministic Dolev-Strong protocol \cite{dolev1983authenticated} solves Byzantine broadcast in worst-case $f+1$ rounds, matching a lower bound~\cite{fischer1982lower}. 
Under asynchrony, Byzantine broadcast is unsolvable even with a single failure. Byzantine reliable broadcast relaxes the termination property of Byzantine broadcast, and the classic Byzantine reliable broadcast by Bracha~\cite{bracha1987asynchronous} has a \latency of 3 rounds and \badlatency of 4 rounds with optimal resilience $n\geq 3f+1$. 
Later works improves the \latency of reliable broadcast to 2 rounds by trading off resilience~\cite{imbs2016trading} or using authentication (signatures)~\cite{abraham2021goodcase}.
A recent line of work studies the \latency of authenticated BFT protocols, including~\cite{synchotstuff, abraham2020brief, abraham2021goodcase}. 
% The result from this paper answers some open questions from~\cite{abraham2021goodcase} about \latency of unauthenticated broadcast protocols.

\section{Conclusion}
In this paper, we investigate the \latency of unauthenticated Byzantine fault-tolerant broadcast, which is time for all honest parties to commit given that the broadcaster is honest. We show the tight results are 2 rounds under $n\geq 4f$ and 3 rounds under $3f+1\leq n \leq 4f-1$ for asynchronous Byzantine reliable broadcast, which can be extended for synchronous Byzantine broadcast as well.
In addition, we also study the \badlatency for asynchronous BRB which measures how fast can all honest parties commit when the broadcaster is dishonest and some honest party commits. We show 2 impossibility results and 4 matching asynchronous BRB protocols, including $(2,4)$-BRB under $n\geq 4f$, F2-BRB of $(2,3)$-round under $n\geq 4f,f=2$, F1-BRB of $(2,2)$-round under $n\geq 4f, f=1$, and $(2,3)$-BRB under $n\geq 5f-1$.

% An interesting question is that, does there exist an $(2,3)$-round BRB protocol with optimal resilience $n\geq 4f$ for $f\geq 3$, or is there an impossibility result?

\section*{Acknowledgments}
The authors would like to thank Kartik Nayak for helpful discussions related to the paper.

%%
%% Bibliography
%%

%% Please use bibtex, 

\bibliography{ref}

\clearpage

\appendix

\section{$3\delta$ Unauthenticated Byzantine Broadcast under Synchrony}
\label{sec:sync:3r}
For completeness, we show an unauthenticated BB protocol in Figure~\ref{fig:3rbb} with \latency of $3\delta$ under synchrony and $n\geq 3f+1$, inspired by Bracha's reliable broadcast~\cite{bracha1987asynchronous}. 

\begin{theorem}\label{thm:sync:3r}
    The protocol in Figure~\ref{fig:3rbb} solves Byzantine broadcast under synchrony with resilience $n\geq 3f+1$ and \latency of $3\delta$.
\end{theorem}
The correctness proof is similar to that of Theorem~\ref{thm:ub:sync}, and we omit it here for brevity.

\begin{figure}[h]
    \centering
    \begin{mybox}
Initially, every party $i$ starts the protocol at most $\delta$ time apart with a local clock and sets $\lock=\bot$, $\clockskew=\Delta$.

\begin{enumerate}%[itemsep=0pt,topsep=0pt]
    \item\label{bb2:step:propose} \textbf{Propose.} The designated broadcaster $L$  with input $v$ sends $\langle \texttt{propose}, v\rangle$ to all  parties.

    \item\label{bb2:step:echo} \textbf{Echo.} 
    When receiving the first proposal $\langle \texttt{propose}, v\rangle$ from the broadcaster,
    a party sends an \echo message for $v$ to all parties in the form of $\langle \echo, v \rangle$.

    \item\label{bb2:step:vote} \textbf{Vote.} 
    \begin{itemize}
        \item When receiving $\langle \echo, v \rangle$ from $n-f$ distinct parties, a party sends a \vote message for $v$ to all parties in the form of $\langle \vote, v \rangle$ and sets $\lock = v$ if not yet sent \vote.
        \item When receiving $\langle \vote, v \rangle$ from $f+1$ distinct parties, a party sends a \vote message for $v$ to all parties in the form of $\langle \vote, v \rangle$ and sets $\lock = v$ if not yet sent \vote.
    \end{itemize}

    \item\label{bb2:step:commit} \textbf{Commit.}
    When receiving $\langle \vote, v \rangle$ from $n-f$ distinct parties at time $t$, a party sets $\lock=v$. If $t\leq 3\Delta+\clockskew$, the party commits $v$.
    
    \item\label{bb2:step:ba} \textbf{Byzantine agreement.}
    At local time $4\Delta+2\clockskew$, a party invokes an instance of Byzantine agreement with \lock as the input. If not committed, the party commits on the output of the Byzantine agreement. Terminate.

\end{enumerate}
    \end{mybox}
    % \vspace{-1em}
    \caption{$3\delta$ unauthenticated BB protocol under synchrony and $n\geq 3f+1$}
    % \vspace{-1em}
    \label{fig:3rbb}
\end{figure}

\end{document}